%% file: main.tex
\documentclass[twoside,leqno,twocolumn]{article}

\usepackage[letterpaper]{geometry}

\usepackage{ltexpprt}

\input{math_commands.tex}

\usepackage{booktabs} % For formal tables
\usepackage{amsmath,amssymb,amsfonts}
\usepackage{algorithm}
\usepackage{algorithmic}
\usepackage{graphicx}
\usepackage{textcomp}
\usepackage{xcolor}
\usepackage{multirow}
\usepackage{subfigure}
\usepackage[super]{nth}

\usepackage{lipsum} % for filler text

\graphicspath{ {img/} }

\begin{document}

\title{\Large Vertex-reinforced Random Walk for Network Embedding}

\author{
    Wenyi Xiao \thanks{Department of Computer Science and Engineering, Hong Kong University of Science and Technology. (wxiaoae@cse.ust.hk)}
  \and
    Huan Zhao \thanks{4Paradigm Inc., China. (zhaohuan@4paradigm.com)}
        \and
    Vincent W. Zheng \thanks{Webank, China. (vincentz@webank.com)}
    \and
    Yangqiu Song \thanks{Department of Computer Science and Engineering, Hong Kong University of Science and Technology. (yqsong@cse.ust.hk)}
}
%\author{Wenyi Xiao, Yangqiu Song\thanks{Hong Kong University of Science and Technology.}
%\and Huan Zhao\thanks{4Paradigm Inc, Beijing, China.}
%\and Vincent W. Zheng\thanks{Webank, China.}}

%\date{Hong Kong University of Science and Technology. 4Paradigm Inc, Beijing, China. Webank, China.}
\date{}
\maketitle

% Copyright Statement
% When submitting your final paper to a SIAM proceedings, it is requested that you include 
% the appropriate copyright in the footer of the paper.  The copyright added should be 
% consistent with the copyright selected on the copyright form submitted with the paper.
% Please note that "20XX" should be changed to the year of the meeting.

% Default Copyright Statement
\fancyfoot[R]{\scriptsize{Copyright \textcopyright\ 2020 by SIAM\\
Unauthorized reproduction of this article is prohibited}}

% Depending on which copyright you agree to when you sign the copyright form, the copyright 
% can be changed to one of the following after commenting out the default copyright statement
% above.

%\fancyfoot[R]{\scriptsize{Copyright \textcopyright\ 20XX\\
%Copyright for this paper is retained by authors}}

%\fancyfoot[R]{\scriptsize{Copyright \textcopyright\ 20XX\\
%Copyright retained by principal author's organization}}

%\pagenumbering{arabic}
%\setcounter{page}{1}%Leave this line commented out.

\begin{abstract} \small\baselineskip=9pt 
% Network embedding has been a

In this paper, we study the fundamental problem of random walk for network embedding.
We propose to use non-Markovian random walk, variants of vertex-reinforced random walk (VRRW), to fully use the history of a random walk path. 
To solve the getting stuck problem of VRRW, we introduce an exploitation-exploration mechanism to help the random walk jump out of the stuck set.
The new random walk algorithms share the same convergence property of VRRW and thus can be used to learn stable network embeddings.
Experimental results on two link prediction benchmark datasets and three node classification benchmark datasets show that our proposed approach \textit{reinforce2vec} can outperform state-of-the-art random walk based embedding methods by a large margin.
% Random walks have been a widely used method for graph-based tasks, such as PageRank computation, or network embedding. Existing random walk on graphs can be generalized as a Markovian process, which is memoryless, i.e., the next state is only dependent on the current state. In this paper, we propose a novel type of random walk on graphs, which is a non-Markovian process, i.e., the next state is dependent on the whole history states. Motivated by the established vertex-reinforced random walk (VRRW), we design a distribution-reinforced random walk (reinforce2vec-d), which can not only converge efficiently to a stationary distribution but also alleviate the inherent ``stuck'' problem of VRRW. By applying reinforce2vec-d to the graph embedding tasks, reinforce2vec-d outperforms existing Markov based random walk methods. Moreover, reinforce2vec-d enjoys the advantages of automatically and efficiently stopping given a graph comparing to the fine-tuned longer walk length of Markov based random walk methods, leading to the huge practical potential for graph-based tasks.
\end{abstract}

% {\small\baselineskip=9pt 

\section{Introduction}

Network representation learning~\cite{HamiltonYL17survey} has been widely studied and used in data mining community to support many applications such as social network mining~\cite{liu2016aligning} and recommendation systems~\cite{EGES,ying2018graph}. 
Network representation learning can be either supervised by a downstream task, e.g., node classification for graph neural networks~\cite{kipf2017semi}, or self-supervised by the adjacency relations, e.g., to approximate a certain proximity defined on the graph, as surveyed in \cite{HamiltonYL17survey}.
When considering the latter case and the network is large, random walks can be applied to improve the scalability.
For example, in DeepWalk~\cite{deepwalk}, it first samples random walk paths from the network following the transition probabilities. Then it performs the Skip-gram algorithm~\cite{mikolov2013efficient} on the random walk paths to learn node embeddings to preserve the truncated average commuting time between nodes~\cite{cao2015grarep,YangSLT17}.
Node2vec~\cite{node2vec} further generalizes this idea to introduce a second-order random walk to balance the breadth and depth of the search to explore neighborhoods of nodes.
Both of them still focus on the Markovian property of the random walk and cannot make full use of the history of the whole path a walker has visited.
It may be possible to use higher-order Markovian chain based random walks to leverage longer histories.
However, this may increase the number of hyper-parameters and sampling complexity.
Nowadays, we have observed many different network representation learning algorithms being developed in the field \cite{HamiltonYL17survey}.
However, when considering the self-supervised learning case for network representation learning, the importance of performing random walk was underestimated.

In this paper, we study the fundamental problem of random walk that can affect the network representation learning.
To leverage the whole visited nodes in a path in a simple way, we propose to use a novel random walk which has a memory to remember the frequency of each visited node and sample based on both adjacent neighbors and the memory. 
This random walk has a strong mathematical foundation based on the vertex-reinforced random walk (VRRW)~\cite{pemantle1992vertex}.
Different from the original VRRW, our random walk is guided by two factors.
The first one is to guide the random walk to follow the memory, which is the same as VRRW.
However, methods that only based on memory will get the random walk stuck in a finite set of vertices~\cite{benaim2011dynamics}.
To resolve this problem, we use the second one to guide the random walk to explore unvisited nodes in neighbors to jump out of the stuck set.
By combining the two factors, it is naturally an exploitation and exploration mechanism widely used in bandit \cite{kocsis2006bandit} or reinforcement learning \cite{MasteringGo2016}.

\noindent$\bullet$ To realize exploitation, we first implement the original VRRW, of which convergence has been proved~\cite{pemantle1992vertex}, which means that our embedding is approximating stationary distributions of hitting time and commuting time.
Based on the convergence analysis, we also propose a new random walk called distribution-reinforced random walk (DRRW).
DRRW is guided by the convergence property of the random walk path evaluated by Kullback-Leibler (KL) or Jensen-Shannon (JS) divergence.

\noindent$\bullet$ To realize exploration, we use both $\epsilon$-greedy and Upper  Confidence  Bounds (UCB)~\cite{kocsis2006bandit} based exploration.
In particular, by $\epsilon$-greedy, with a probability $\epsilon$, we randomly select the neighbors without any bias.
With UCB, it is more likely explore a node's neighbor that is less visited.
By annealing the exploration factor, we can also make the walking path converge to a stationary distribution as the original VRRW.

To make a fair evaluation comparing different random walk models, we use the problem of standard network embedding as the test-bed.
We use two datasets for link prediction and three datasets for node classification to evaluate different algorithms. Experimental results show that our random walk based methods can outperform existing state-of-the-art models by a large margin.
This verifies our intuition that random walk plays a crucial role in network representation learning and should be paid more attention.
The contributions of the work can be summarized as follows:
\begin{enumerate}
     \item We propose to use a reinforced random walk, a non-Markovian process looking at the whole history of a path, to replace the original Markovian random walk used for network embedding and prove that the proposed random walks have stationary distributions to learn stable embeddings. The framework is denoted as \textit{reinforce2vec}.
     %For the best of our knowledge, we are the first to consider the change of transition probability over time for graph embedding, while previous works always regarding it as constant variables.
     \item We extend the original VRRW by using KL and JS divergences to guide the random walk, which is denoted as distribution-reinfoced random walk (DRRW). To alleviate the stuck set problem of VRRW and DRRW, we also propose an exploration mechanism to jump out of the stuck set.
     %We propose a reinforced random walk in an exploitation-exploration manner. This time-variant random can balance the convergence and the diversity of the random walk generated paths. 
     \item We conduct extensive experiments on real-world datasets, demonstrating that the proposed reinfoce2vec can improve the embedding performance compared to existing state-of-the-art random walk based network embedding methods. The code is available at https://github.com/HKUST-KnowComp/vertex-reinforced-random-walk.
\end{enumerate}

\section{Related Work}

Random walks have been widely used for graph and network data mining. 
For example, PageRank~\cite{page1999pagerank} is one of the earliest data mining algorithms that proven to be useful for real Web applications.
The stationary distribution of the first-order Markov chain of the random walk can be used as the rank authorities of Web pages.
Later, personalized PageRank has been proposed~\cite{JehW03,Haveliwala03} and is further applied to semi-supervised learning~\cite{ZhouS04}.
The label propagation based on personalized PageRank can be analyzed by the expected commuting time~\cite{HamLMS04} between nodes~\cite{ZhouS04}.
Vertex-reinforced random walk (VRRW) was studied originally by the probability community \cite{pemantle1992vertex}.
It was proven that the random walk can converge to a stationary distribution~\cite{pemantle1992vertex}.
Thus, later this idea has been introduced into PageRank \cite{mei2010divrank} and multilinear or higher-order PageRanks~\cite{GleichLY15,BensonGL17}. VDRW \cite{wu2018imverde} and HeteSpaceyWalk \cite{he2019hetespaceywalk} learn node representations in imbalanced Network and Heterogeneous Network, respectively.
However, as pointed out by~\cite{benaim2011dynamics}, VRRW can easily get stuck to a finite set of vertices, which may not be good enough to explore the full graph based on random walks.

Random walks are also widely used in network embeddings.
The representative works are DeepWalk~\cite{deepwalk} and node2vec~\cite{node2vec}, which use first-order and second-order random walks to guide the network embedding respectively.
Later a lot of analysis was conducted to build the relationship between the random walk based methods and the matrix factorization based methods, which shows the fact that DeepWalk uses the embedding vectors to learn from the truncated average commuting time between nodes~\cite{cao2015grarep,YangSLT17}.
This also implies that when performing random walk, to obtain a set of stable embedding vectors, one needs to run sufficient steps of random walk to estimate the stationary distributions of hitting time and commuting time.
However, there is a lack of study or analysis for this issue of using random walk.
Expected commuting time can also be used to learn network embedding, where a personalized random walk should be introduced to  obtain the stationary distributions~\cite{ChoBP15}.
In theory, the stationary distribution of a first-order Markov chain based random walk can be learned, however, in practice, we still need to perform truncated approximation to make the algorithm efficient~\cite{ChoBP15} (more details are in the arxiv version). 

Different from the existing random walk based embedding methods, our VRRW based algorithms enjoy two advantages.
First, being a vertex-reinforced random walk, the whole history of the random walk path can be leveraged.
Although being non-Markovian, it can be used to approximate high-order PageRanks~\cite{GleichLY15,BensonGL17}.
Second, as the stationary distribution of each random walk path can be guaranteed, we do not need to worry about the stationary distributions of hitting and commuting times. 
The embeddings learned from our random walk path are guaranteed to be stable.

\section{Reinforce2vec}

In this section, we introduce our proposed reinforced random walk based algorithms for network embedding. The notations in the remaining sections and their descriptions are shown in Table~\ref{tb:notations}.

Given an unweighted (un)directed network $G = (V, E)$, with vertices $V = \left \{ v_1, ..., v_N \right \}$  and edges $E = \left \{ e_1, ..., e_M \right \}$, the goal of network embedding is to determine a set of fixed length vectors $\vz_t \in \sR^d$ for each vertex $v_t$ such that similar vertices are close in the representation space. 
% In other words, the purpose of network embedding is to learn a mapping function $\phi$: $V \mapsto \sR^d$.
% Introduced by \cite{deepwalk}, a modeling of the vertex representation encodes the node as a function of its co-occurrences with other nodes in finite truncated random walks. These co-occurrences capture the diffusion in the neighborhood around each vertex in the graph and approximate the local community structure around a node. 
Precisely, the goal of the embedding method is to learn a representation that enables an estimate of the likelihood of the anchor node $v_t$ co-occurring with its neighbors, $v_x\in\left \{v_{t-C}, ...,v_{t-2}, v_{t-1}, v_{t+1}, v_{t+2}, ..., v_{t+C} \right \}$, the context nodes. In particular, the Skip-gram model \cite{mikolov2013efficient} models the conditional probability of a vertex pair, $(v_t, v_x)$ in the range of window size by a log-linear function of the inner product between the vectors $\vz_t$ and $\vz_x$ as follows:
\begin{equation}
\label{ucb}
    p(v_x|v_t) = \frac{\text{exp}(\vz_t\cdot \vz_x)}{\sum_{v_k \in V}\text{exp}(\vz_t \cdot \vz_k)}.
\end{equation}
The representation vector $\vz_k$ for each vertex $v_k \in V$ can be found by minimizing the cross-entropy loss function.
% Notice that the exact computation of this conditional probability is computationally intractable for increasing lengths of the random
% walks, so in particular \cite{deepwalk} use a hierarchical softmax to approximate the computation. 

\begin{table}[t]
\caption{Notations.}
\resizebox{0.45\textwidth}{!}{
\begin{tabular}{c|c}
\toprule
Notation & Description \\ 
\midrule
$G$ & Input network \\
$V$ & Vertext set of $G$ with $\left | V \right | = N$ \\
$E$ & Edge set of $G$ with $\left | E \right | = M$ \\
$v_i$ & A node $v_i \in V$ \\
$e_{ij}$ & Edge $e_{ij} \in E$ connecting $v_i$ and $v_j$ \\
$d$ & Embedding dimension \\
$R$ & Walks per node \\
$L$ & Walk length \\
$C$ & Window size for Skipgram \\
%$\lambda$ & Trade-off to balance Exploitation \& Exploration \\
$\epsilon$ & A probability we take a random action in $\epsilon$-Greedy\\
$X(s)$ & The $s$-th node in the walk \\
%$\pi _{tx}$ & Transition probability from node $t$ to $x$ \\
$\vz_i$ & The embedding vector of $v_i$ \\
$\vw(n)$ & The occupation vector at step $n$ \\
%$N_S(u)$ & Neighbors of node $u$ with sampling strategy $S$ \\ 
\bottomrule
\end{tabular}
}
\label{tb:notations}
\end{table}

%\huan{Better to introduce formally graph embedding and random walk.}

\begin{figure*}[t]
\centering
\includegraphics[width=0.8\textwidth]{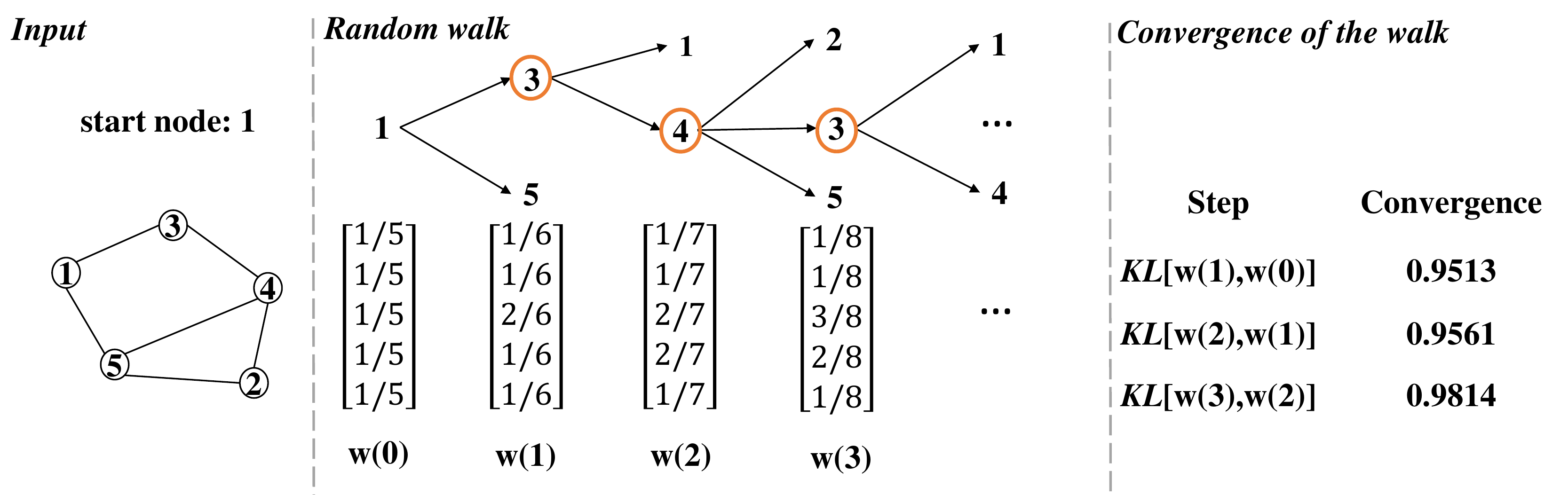}
\vspace{-0.15in}
\caption{\footnotesize We illustrate DRRW by this toy example. The inputs are a graph and a start node ``1''. The initial path is (``1''). The initial occupation vector $\vw(0)$ records the node distribution during the random walk process. When selecting the next node from "1", we have two neighbors: ``3'' and ``5''.  Suppose we add ``3'' to the current path, then we get a new occupation vector $\vw^3(0)$. Here, we calculate the KL divergence between $\vw(0)$ and $\vw^3(0)$ to get the exploitation score for ``3''. We get the exploitation score for ``5'' by the same process. We compare the two exploitation scores and choose the one with the maximum score. Since the exploitation scores for ``3'' and ``5'' are the same, we randomly select the next node and jump to ``3''. Now, the path becomes (``1''$->$ ``3'') and we update the occupation vector to $\vw(1)$. We then jump to ``4'' with random selection for the same reason, updating the path and the occupation vector. Looking at ``4'' 's neighbors: ``2'', ``3'', ``5'', we add ``3'' to the path as it can get the highest exploitation score. We select each neighbor by considering both the current node and the latest occupation vector step-by-step until the walk length reaches to $L$. We make use of the whole history during the walking process. The right part shows the convergence results.} 
\label{fg:drrw-exa}
\vspace{-0.1in}
%\label{fg-reinforcewalk}
\end{figure*}

\subsection{ The General Framework.}
%\huan{Briefly introduce the whole thing, including motivation, outline, etc.}
We sample the neighborhoods of a source node through a strategy called Exploitation-Exploration trade-off, motivated by \cite{xiao2019beyond}. The strategy has been most thoroughly studied through the multi-armed bandit problem and for finite state space MDPs in \cite{burnetas1997optimal}.
A brief introduction of the strategy is that it models an agent that simultaneously attempts to optimize the decisions given the current information (called ``exploitation'') and gather more information that might lead us to better decisions in the future (called ``exploration''). 
Specifically, at each step, we select the next node by combining the exploitation process based on VRRW or DRRW and exploration process based on $\epsilon$-Greedy or UCB.
%In addition, exploitation is related to local search while exploration is related to a global search. In the first one, we want to refine the solution and try to avoid big jumps on the search space, whereas in the second one we are interested in exploring the search space looking for good solutions.
In the following section, we will give the details on how we design the two processes and keep a balance of exploitation and exploration.

\subsection{Exploitation.}
This process tends to exploit the seemingly optimal neighbor, specifically selecting those with higher rewards by certain designed approaches. Here we give two approaches to calculate the rewards for each candidate neighbor, vertex-reinforced random walk (VRRW) and its extension, distribution-reinforced random walk (DRRW).
%We present exploitation reward, one is VRRW and another is DRRW
% \subsubsection{VRRW.}
\subsubsection{Vertex-Reinforced Random Walk.}
Vertex-reinforced random walk (VRRW), defined by~\cite{pemantle1992vertex}, is a class of non-Markovian discrete-time random processes on a finite state space. 
%Let $G$ be any homogeneous graph without loop. 
For any $x \in G$ and $V \subset G$, denote $x \sim V$ if there exists a vertex $v \in V$ such that $x \sim v$. Here $\sim$ represents the neighbor relation. For the process consists of a sequence of nodes $\mathcal{F}_n = X(0), X(1), X(2), ... X(n)$, VRRW defines the local time 
\begin{equation}
    Z(n, v) = 1+\sum_{s=1}^{n}\delta [ X(s) = v]
\end{equation}
to be the number of times the vertex $v$ has been visited by step $n$, plus 1.
$\delta[\cdot]$ denotes an indicator function where $\delta[true]=1$ and $\delta[false]=0$.
The transition probability for VRRW is defined as:
\begin{equation}
    P(X(n+1)=x | \mathcal{F}_n) = \delta[{x \sim X(n)}]\frac{Z(n,x)}{\sum_{y \sim X(n)}Z(n,y) },
\end{equation}
where $\delta[{x \sim X(n)}]$ indicates that we only count the probability when $x$ is a neighbor of $X(n)$.
% \begin{equation}
%     P(X(n+1)=x | \mathcal{F}_n) = \mathbf{1}^T\mathbf{1}_{x \sim X(n)}\frac{Z(n,x)}{\sum_{y \sim X(n)}Z(n,y) },
% \end{equation}
%where $\mathbf{1}$ is an all-one vector and $\mathbf{1}_{x \sim X(n)}$ is an indicator vector where 1 indicates the corresponding indexed node is a neighbor of $x$.
Therefore, the probability of a move  to a neighbor $x$ is proportional to the local time at $x$ at step $n$. We regard this local time as exploitation score $Q_x(n)$.
% where $\mathcal{F} = \phi(X(1), X(2), ...,X(n)$. 
% Therefore, a move to a neighbor $x$ is proportional to the local time at $x$ at step $n$.

\subsubsection{Distribution-Reinforced Random Walk.}
Distribution-Reinforced Random Walk (DRRW) is an extension of VRRW. An example is shown in Figure \ref{fg:drrw-exa}. The major difference between the two methods is that we do not simply make the walk intend to select those nodes visited before. Instead, the path would select those which can make the distribution of nodes in the path converge more quickly. The process of DRRW is also a non-Markovian process. We give a theoretical theorem and proof for the reason why we can also calculate the reward by measuring the convergence.

Before that, to measure the node distribution in the path, we use the fractional occupation vector, $\vw(n)$, which is obtained by processing the normalization of each node frequency in the path up to the step $n$, as shown in:
\begin{equation}
\label{eq:ov}
   \vw_i(n) = \frac{1}{n+N}\left (  1+\sum_{s=1}^{n}\delta [ X(s) = v_i ] \right ).
\end{equation}

\begin{theorem}
\label{thm:converge}
If the random walk process $\left \{ X(n) \right \}$ converges to a unique stationary distribution $\vx(n)$, it must converge to a point where $\vx(n) = \vw(n)$.
\end{theorem}

\begin{proof}[Proof of Theorem \ref{thm:converge}]
%\huan{I think it is better to refer to the proof in the original paper, say like Following the proof in XXX. We are actually copying the proofs.}
    Consider the behavior of the random walk process at step $n \gg 1$ and some time $n + L\gg n$ where $n > L \gg 1$. If the process $\left \{ X(n) \right \}$ converges to a unique stationary distribution $\vx(n)$, we have:
\begin{equation}
\label{eq:kl}
\begin{aligned}
    \vw(n+L) &\approx \frac{n\vw(n)+L\vx(n)}{n+L} \\
    &= \vw(n) + \frac{L}{n+L}(\vx(n)-\vw(n)).
\end{aligned}
\end{equation}

In a continuous time limit $L\rightarrow 0$, we have:
\begin{equation}
\label{eq:js}
\begin{aligned}
    \frac{d\vw(n)}{dL} &\approx \lim_{L\rightarrow 0}\frac{\vw(n+L)-\vw(n)}{L} \\
    &= \frac{1}{n}\left ( \vx(n) - \vw(n) \right ).
\end{aligned}
\end{equation}

Hence, the convergence of the random walk to a fixed point is equivalent to the convergence of the occupation vector $\vw$ for the stochastic process. 
\end{proof}

%The proof follows the spacey random walk theory \cite{BensonGL17}. 
For every step, our aim is to select the neighbor which is able to maximize the possibility of the occupation vector $\vw$ to converge. Here, we give two methods to measure the convergence between two distributions, $\vw(n)$ for current path and $\vw^x(n+1)$ for adding each neighbor node $x$ to the current path.
\begin{itemize}
    \item \textbf{Kullback-Leibler Divergence}
            \begin{equation}
            \begin{aligned}
            Q_x(n)' &=   D_{KL}(\vw(n)\left |  \right | \vw^x(n)) \\
            &=  \sum_{i=1}^N \vw_i(n) \text{log}(\frac{\vw_i(n)}{\vw^x_i(n)}). 
            \end{aligned}
            \end{equation}
    \item \textbf{Jensen-Shannon Divergence}
            \begin{equation}
            \begin{aligned}
             Q_x(n)' &=  \frac{1}{2}D_{KL}(\vw(n)\left |  \right | \vm) \\
             &+\frac{1}{2}D_{KL}(\vw^x(n)\left |  \right | \vm)
             \end{aligned}
            \end{equation}
            where $\vm = \frac{1}{2}\left [ \vw(n) + \vw^x(n) \right ]$.
\end{itemize}
% \huan{references for KL and JS.}
The lower the divergence value, the better we have matched the next distribution with the current distribution. Therefore, the exploitation score is: $Q_x(n) = 1- Q_x(n)'$.

% \subsubsection{Vertex-Reinforced Random Walk.}

%\subsubsection{Distribution-Reinforced Random Walk.}

To sum up, we tend to select those that achieve higher exploitation scores by VRRW or DRRW. 

\subsection{Exploration.}

\begin{figure}[t]
\centering
\includegraphics[width=0.49\textwidth]{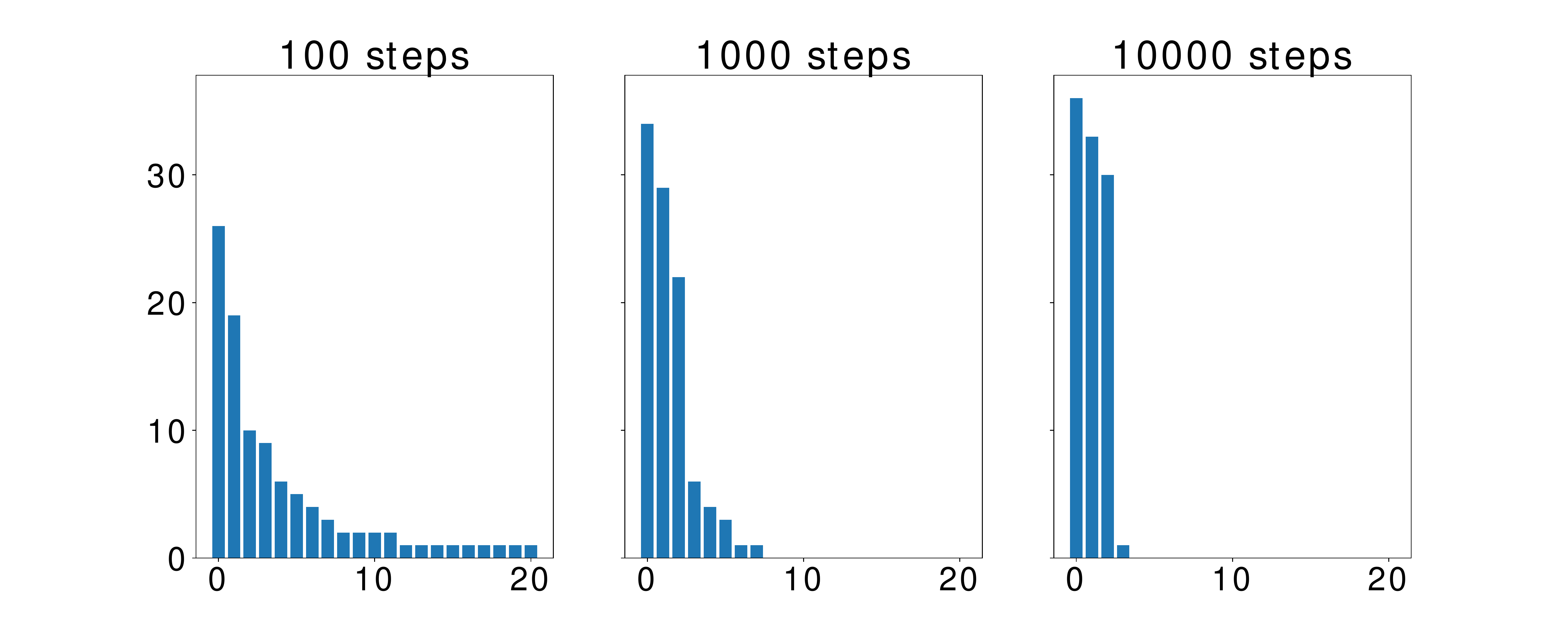}
\vspace{-0.3in}
\caption{\footnotesize An example for node frequency of the last 100 steps after 100, 1000, 10000 steps on the PPI dataset. The $x$ axis denotes the number of existing nodes (not the real node IDs) and the $y$ axis denotes the frequencies of each node. We can see after 10000 steps, only 4 nodes appeared in the last 100 steps and three nodes occupy mostly. It indicates that the random walk gets stuck to 3 nodes.}
\vspace{-0.1in}
\label{fig:node-dis}
\end{figure}

\label{sec-exploration}
On arbitrary graphs, \cite{volkov2006phase} proves that VRRW localizes with positive probability on some specific finite subgraphs.
\begin{theorem}
\label{the-trapping}
Let $T$ be a strongly trapping subset of $(G, \sim)$. Then the VRRW has asymptotic range $T$ with a positive probability.
\end{theorem}
We give a example of random walk with VRRW running 100, 1,000 and 10,000 steps, shown in Figure \ref{fig:node-dis}. It further proves the getting stuck problem.
Therefore, it is also needed to do the exploration, taking different actions to skip out of this trap.
We design the exploration mechanism in two ways: $\epsilon$-Greedy Algorithm and Upper Confidence Bounds.

\subsubsection{$\epsilon$-Greedy Algorithm.}
The main idea of the $\epsilon$-Greedy algorithm is to take the best action most of the time, but do random exploration occasionally. 
In detail, with a probability $\epsilon$ we randomly select the neighbors without any bias. Meanwhile, with probability $1-\epsilon$ we select neighbors which are visited more often before (VRRW) or can make the node distribution in the path converge more quickly (DRRW).

\subsubsection{Upper Confidence Bounds.}
In the UCB algorithm, we always select the greedy action to maximize the upper confidence bound.
Here we use the upper confidence bounds one (UCB1)~\cite{kocsis2006bandit} value to determine the next move from a viewpoint of multi-armed bandit problem. The selection strategy is defined by: 

\begin{equation}
\label{eq:ucb}
    s_x(n) = Q_x(n) + U_x(n),
    %&a = \argmax_v (s(v)),
\end{equation}
where $Q_x(n)$ denotes the exploitation reward of node $x$ at step $n$ and $U_x(n)$ is the exploration score of node $x$ at step $n$. 
% \huan{If you remove $\lambda$, you need to remove it in Table~\ref{tb:notations}.}

For our random walk process, we can add the exploration score $U_x(n)$ in order to push the path to explore new nodes. We design the exploration mechanism to get $U_x(n)$ for node $v$ as follows:
        \begin{equation}
        U_x(n) = \sqrt{\frac{\text{log} (Z(n, u))}{Z(n, x)}},
       \label{eq-ucb-u}
        \end{equation}
where $Z(n, u)$ and $Z(n, x)$ are denoted as 1 plus the times that the start node $u$ and the neighbor node $x$ of current node $X(n)$ have been visited up to step $n$, respectively. 
%\huan{explain $\lambda$}

%This equation clearly expresses the exploration-exploitation trade-off: while the first term of the sum tends to exploit the seemingly optimal arm, the second term of the sum tends to explore less pulled arms. $\lambda$ is used to balance the two terms.

%\begin{itemize}
%    \item Firstly, inspired by the exploring process in reinforcement learning, we follow the random walk in DeepWalk %\cite{deepwalk} with a high probability at the beginning, purely selecting the neighbors without any bias. With the steps %increasing, the probability of following DeepWalk process decays and the path tends to select neighbors which could make the %distribution converge more quickly.
%    \item Second method to explore is that we can add the explored reward $U_v(n)$ in order to push the path to explore new %states/nodes. We design the exploration mechanism to get the explored reward $U_n(v)$ for node $v$ as follows:
%        \begin{equation}
%        U_v(n) = \sqrt{\frac{\text{log} (T_u(v))+1}{T_v(n)+ 1}},
%       \label{eq-ucb-u}
%        \end{equation}
%        where $T_u(n)$, $T_v(n)$ denote as the times that the source node $u$ and the current node $v$ have been selected  at %search step $n$, respectively.
%\end{itemize}

%Exploration is also needed because there is always uncertainty about the accuracy of the action-value estimates. 
%The greedy actions are those that look at the present, but some of the other actions may be better. Therefore, sometimes it’s worth to take different actions to get new data.

\textbf{Continuous Approximation as Transition Probability.}
After calculating the exploitation and exploration scores for each neighbor, a natural way to choose the next node in the path is to pick the neighbor with the maximum sum score by Eq.~(\ref{eq:ucb}).  
Following this way, we can only get one path no matter the paths number increases since it is a deterministic approach. 
It is easy to suffer from local optimal or stuck in a loop for both VRRW and DRRW. 
To handle this, we sample the next node with a continuous transition probability by adding a softmax layer as:
\begin{equation}
\label{eq-trans}
    P(X(n+1)=x | \mathcal{F}_n) = \frac{\text{exp}(s_x(n))}{\sum_{y \sim X(n)} \text{exp}(s_y(n))}.
\end{equation}
This process can also be regarded as a way to balance of exploration and exploitation since picking the neighbor with the maximum sum score simply tends to the direction for the highest current reward while using the softmax function to make transition probability enables the path to explore diverse states.

\begin{algorithm}[t]
    \caption{The reinforce2vec algorithm.} \label{alg-1}
    {
        \begin{algorithmic}
            \STATE {\bf LearnEmbeddings:} (Graph $G = (V, E)$. Dimensions $d$, Walks per node $R$, Walk length $L$, Window size $C$. )
            \STATE Initialize walks to Empty
            \WHILE {Iteration $r < R$}
            {
                \FORALL {nodes $u \in V$}
                    \STATE $walk = \text{DRRW}(G, u, L)$
                    \STATE {Append $walk$ to $walks$}
                \ENDFOR
                
            }
            \ENDWHILE
        \STATE {$\left \{ \vz_i \right \}$ = Skip-gram $(C, d, walks)$ }
        \end{algorithmic}
    
    }\vspace{-0.01in}
\end{algorithm}

\begin{algorithm}[t]

    \caption{DRRW.} \label{alg-2}
    {
        \begin{algorithmic}
            \STATE {\bf LearnEmbeddings:} (Graph $G = (V, E)$. Start node $u$. Walk length $L$. )
            \STATE Initialize walks to [u]
            \STATE Initialize occupation vector $\vw$ by Eq.~\ref{eq:ov}
            \FORALL {$walk\_len = 1$ to $L$}
                \STATE $curr = walk[-1]$
                \STATE $N_{curr} = GetNeighbors(cur, G)$
                \STATE Get transition probabilities $P_{curr}$ from Eqs.~(\ref{eq:kl}) $\sim$ (\ref{eq-trans})
                \STATE $s$ = Sample($N_{curr}$, $\pi_{curr}$, $\vw$)
                \STATE Append $s$ to $walk$
                \STATE Update $\vw$
                %\STATE {Append $walk$ to $walks$}
            \ENDFOR
        \end{algorithmic}
    
    }\vspace{-0.01in}
\end{algorithm}

\subsection{Algorithm.}
%We denote the source node $u$ as starting node as $c_0$ and denote $c_n, l\in \left [ 1, L \right ]$ as the $n$-th node added for $u$ in the path.
% At search step $n$, the node $c_{n+1}$ is retrieved from the neighbors of node $c_n$. We calculate the transition probability according to Eq.~\ref{eq:ucb} and applying it to softmax.

%where $Q_v(n)$ denotes the  exploitation reward of node $v$ and $U_v(n)$ is the utility to explore node $v$ at time $n$. As we can see, $g(Z(n,v)) = Q_v(n)$ while $U_v(n)$ is added to change the distributed-reinforced process for certain purpose which will be explained in Section \ref{sec-exploration}. 
The pseudocode for reinforce2vec based on DRRW with UCB is given in Algorithm~\ref{alg-1} and the details of the designed random walk for every single path are shown in Algorithm~\ref{alg-2}.
%We initialize paths for users by randomly selecting users in the social graph and set users' explored times by the times they selected as friends. For each user in the $t$-th day training, we first explore $B$ sets of friends by MCTS strategy and update the explored times for these selected friends, consequently updating the exploration values $U_t(v)$ of them. These $B$ sets of friends are incorporated with the user as input to the RS model. Then we update $Q_t(u)$ of the target user. The updated results are used for the $t+1$-th day training.

\subsection{Complexity Analysis.}
For the random walk phase in a network $G = (V, E)$ with $R$ and $L$, the reinforced walk for each step takes $O(L)$ because for each step we can divide the neighbors into two groups: those visited in the current path and those not. We compute the exploration and exploitation scores for those visited. This process is linear to the maximum path length $L$. For those not visited, the exploration and exploitation scores are the same. We only need to do the computation once. The complexity of softmax layer is $O(L)$. Thus, the total walk complexity is $O(R \times L^2)$. For the training phase, we use a Skip-gram model to train the embeddings, which also has linear complexity and can be parallelized by using the mechanism as in word2vec \cite{mikolov2013efficient}. Overall, the proposed random walk is scalable for a large-scale graph.

\section{Experiment}
In this section, we evaluate the proposed methods by two downstream tasks: link prediction and node classification.

% to answer the following questions:
% \begin{enumerate}
%     \item Can the proposed non-Markovian random walk based graph embedding approach outperform the baselines?
%     \item Is it effective to apply the exploitation-exploration strategy to solve the "stuck" problem of VRRW?
% \end{enumerate}

\subsection{Experimental Settings.} 
\subsubsection{Baselines.}
\label{sec-setting}
We compare our model with popular approaches, including first-order random walk (DeepWalk~\cite{deepwalk}, LINE~\cite{tang2015line}),
second-order random walk (node2vec~\cite{node2vec}), and matrix factorization approaches (GraRep and HOPE~\cite{cao2015grarep, ou2016asymmetric}).
We denote our algorithms as reinforce2vec-v for embedding with VRRW and reinforce2vec-d for embedding with  DRRW. If we do not mention, reinforce2vec is reinforce2vec-d, which uses DRRW with JS divergence as well as UCB exploration.
%There are many other network embedding methods, but we do not consider them here, because their performance is inferior to these baseline models as shown in corresponding papers. 

%The descriptions of these models are as follows.
%\begin{itemize}
%    \item DeepWalk \cite{deepwalk}: DeepWalk first transforms the network into node sequences by truncated random walk, and then uses it as input to the Skip-gram model to learn representations.
%    \item LINE \cite{tang2015line}: LINE can preserve both first-order and second-order proximities for undirected graph through modeling node co-occurrence probability and node conditional probability.
%    \item GraRep \cite{cao2015grarep}: GraRep preserves high-order proximities by constructing different k-step probability transition matrix.
%    \item node2vec \cite{node2vec}: node2vec develops a biased random walk procedure to explore the neighborhood of a node, which can strike a balance between local properties and global properties of a network.
%    \item HOPE \cite{ou2016asymmetric}: HOPE preserves asymmetric transitivity by approximating high-order proximity.
% \end{itemize}
% We follow the original settings argued in papers. 
%We exclude GraRep~\cite{cao2015grarep}, AttentionWalk~\cite{abu2018watch} since they are unable to efficiently scale to large networks.
Note that, in the experiments, the network embeddings are learned with only the graph structure (nodes and edges), without observing node features nor labels during training. 
All baselines and our model belong to unsupervised network embedding. 
Therefore, those semi-supervised learning approaches like GCN \cite{kipf2017semi} are not considered for comparison.

\subsubsection{Parameter settings.}  
For DeepWalk and node2vec, windows size $C=10,R=80, L=40$.
For node2vec, $p$, $q$ are chosen from $\left \{ 0.25, 0.50, 1, 2, 4 \right \}$.
For LINE, the size of negative-samples is $K = 5$. 
For HOPE, $\beta$ is set to 0.01.
Other parameters follow the original settings in corresponding papers. 
For reinforce2vec, we set $\epsilon = 0.5, R=80, L=40, C=10, d=64$. 
%All experiments were conducted on Linux (CentOS release 6.9), Python 3.7 from Anaconda 4.6.14.

\subsection{Link Prediction Task.}
Link prediction is a challenging task applied in many areas like information retrieval, recommendation system, and social networks. It is used to evaluate the structure-preserving properties of embedding. We compare our method with baselines in the manner of \cite{node2vec}: randomly remove a fraction $(=50\%)$ of graph edges while ensuring that the residual network is connected and then learn embeddings from the remainings. 
We study several binary operators in Table \ref{tb:operations} to construct features for an edge based on its two node vectors, and then the operated feature vector for an edge is as the input to a logistic regression classifier.
\begin{table}[t]
\centering
\caption{Binary operations for learning edge features. The definitions correspond to the $i$-th component of vector $\vu$ and $\vv$.}
\label{tb:operations}
\resizebox{0.45\textwidth}{!}{%
\begin{tabular}{c|c|c|c|c}
\toprule
Operation & Average & Hadamard & Weighted-L1 & Weighted-L2 \\
\midrule
Definition & $(\vu_i + \vv_i)/2$ & $\vu_i \ast \vv_i$ & $\left | \vu_i -\vv_i \right |$ & $\left | \vu_i -\vv_i \right |^2$\\
\bottomrule
\end{tabular}
}
\end{table}

\vspace{-0.2in}

\subsubsection{Datasets.}
For the link prediction task, we use two public datasets: Facebook and ca-AstroPh as follows, details shown in Table \ref{tb:lp-data}. 
\begin{itemize}
    \item Facebook~\cite{leskovec2014snap}: Facebook dataset is a social network representing a friendship relation between any two users.
    %\item wiki-vote: The network contains all the Wikipedia voting data from the inception of Wikipedia. Nodes in the network represent Wikipedia users and directed edges represent voting actions by users.
    \item ca-AstroPh  ~\cite{leskovec2014snap}: This is a collaboration network of Astrophysics which is generated from papers submitted to the arXiv.
\end{itemize}
% Table \ref{tb:lp-data} describes the statistics of the datasets. 

\begin{table}[t]
    \centering
    \caption{ Datasets for link prediction.}
\begin{tabular}{c|cc}
\toprule
Dataset & Facebook & ca-AstroPh \\
\midrule
$|V|$ & 4,039 & 17,903 \\
$|E|$ & 88,234 & 197,031 \\
Nodes Types & users & researchers \\
Edges Types & friendship & co-authorship \\
\bottomrule
\end{tabular}
    \label{tb:lp-data}
\vspace{-0.2in}
\end{table}

\begin{table*}[t]
\centering
\caption{Area Under Curve (AUC) scores for link prediction. Comparison with baselines and our method (reinforce2vec) using binary operators: Hadamard, Average, 
Weighted-L1 and Weighted-L2. The definition is shown in Table \ref{tb:operations}.}
\resizebox{1\textwidth}{!}{%
\begin{tabular}{c|cccc|cccc}
\toprule
 & \multicolumn{4}{c|}{Facebook} & \multicolumn{4}{c}{ca-AstroPh} \\
 \midrule
Operation & Hadamard & Average & Weighted-L1 & Weighted-L2 & Hadamard & Average & Weighted-L1 & Weighted-L2 \\
\hline
Deepwalk & 96.8 & 72.4 & 95.7 & 95.8 & 93.4 & 70.7 & 82.8 & 83.1 \\
Node2vec & 96.8 & 72.7 & 96.0 & 96.1 & 93.7 & 72.2 & 84.7 & 84.8 \\
LINE & 95.9 & 70.3 & 93.8 & 94.1 & 89.0 & 65.2 & 88.1 & 88.6 \\
HOPE & 94.9 & 42.1 & 92.8 & 94.3 & 93.1 & 34.6 & 89.2 & 90.8 \\
Grarep & 94.7 & 39.8 & 76.0 & 84.7 & 87.7 & 41.3 & 87.5 & 83.4 \\
reinforce2vec & 95.3 & 56.1 & 98.9 & \textbf{98.9} & 84.0 & 63.3 & 97.4 & \textbf{97.5} \\
\hline
Gain of reinforce2vec$\left [ \% \right ]$ & \multicolumn{4}{c|}{2.2} & \multicolumn{4}{c}{4.1} \\ 
\bottomrule
\end{tabular}
}
\label{tb:lp-result}
\vspace{-0.2in}
\end{table*}

\begin{table}[t]
\centering
\caption{Performance for different variants of VRRW for link prediction on Facebook dataset.}
\resizebox{0.45\textwidth}{!}{%
\begin{tabular}{c|c|c|c|c|c}
\toprule
 \multicolumn{2}{c|}{Variant} & Hadamard & Average & l1 & l2  \\
 \midrule
\multirow{3}{*}{reinforce2vec-v} & w/o exploration & 70.6 & 53.9 & 72.8 & 73.8 \\
 & $\epsilon$-Greedy & 94.3 & 59.2 & 95.4 &  95.7 \\
 & UCB & 96.2 & 54.9 & 96.2 &  96.4  \\
 \hline
\multirow{3}{*}{reinforce2vec-d:KL} & w/o exploration & 66.8 & 51.2 & 62.5 & 63.4  \\
 & $\epsilon$-Greedy & 93.8 & 56.0 & 98.3 & 98.4  \\
 & UCB & 94.5 & 53.8 & 98.3 & 98.4  \\
 \hline
\multirow{3}{*}{reinforce2vec-d:JS} & w/o exploration & 70.0 & 53.5 & 61.0 & 61.7  \\
 & $\epsilon$-Greedy & 96.1 & 63.5 & 98.6 &  98.7 \\
 & UCB & 95.3 & 56.1 & 98.9 & \textbf{98.9}   \\
 \bottomrule
\end{tabular}
}
\vspace{-0.1 in}
\label{tb:lp-result-ablation}
\end{table}

\vspace{-0.1in}

\subsubsection{Performance Evaluation.}
We show the values of AUC (Area Under Curve) scores (one for each binary operator) in Table \ref{tb:lp-result}. 
We have the following observations: 1)  HOPE, GraRep, and LINE perform relatively worse in link prediction, as they are not good at capturing the pattern of edge existence in graphs. 
2) DeepWalk and node2vec are better than other baselines. This is probably because DeepWalk and node2vec are both random walk based models, which are better at extracting proximity information among vertices. 
3) reinforce2vec outperforms all the baselines in link prediction. Specifically, reinforce2vec improves AUC scores on Facebook and ca-AstroPh by
2.2\% and 4.1\% compared with the best baselines, respectively. 
%Besides, it proves the correctives and effectiveness of applying time-variant transition probability for random walk, compared with previous works which all regard it as a constant variable.

In addition to comparison with baselines, we analyze the performance for different vertex-reinforced random walk based methods on Facebook dataset. 
The results are shown in Table~\ref{tb:lp-result-ablation}.
Here we have the following observations: 1) All DRRW based methods (reinforce2vec-d:KL, reinforce2vec-d:JS) perform better than VRRW based method (reinforce2vec-v) under different exploration strategies. This indicates that it is reasonable and more effective to directly reinforce the convergence of the path during random walk than just use the frequency of each node. 
2) Comparing methods with or without exploration, it indicates that the designed exploration ways ($\epsilon$-Greedy and UCB) can effectively help the random walk process skip out of the trap. Consequently, the whole performances improve a lot after handling the getting stuck problem. 
3) For different ways of calculating convergence in DRRW, the performances of KL divergence and JS divergence are comparable.

\begin{figure*}[t]
\centering
\includegraphics[width=0.95\textwidth]{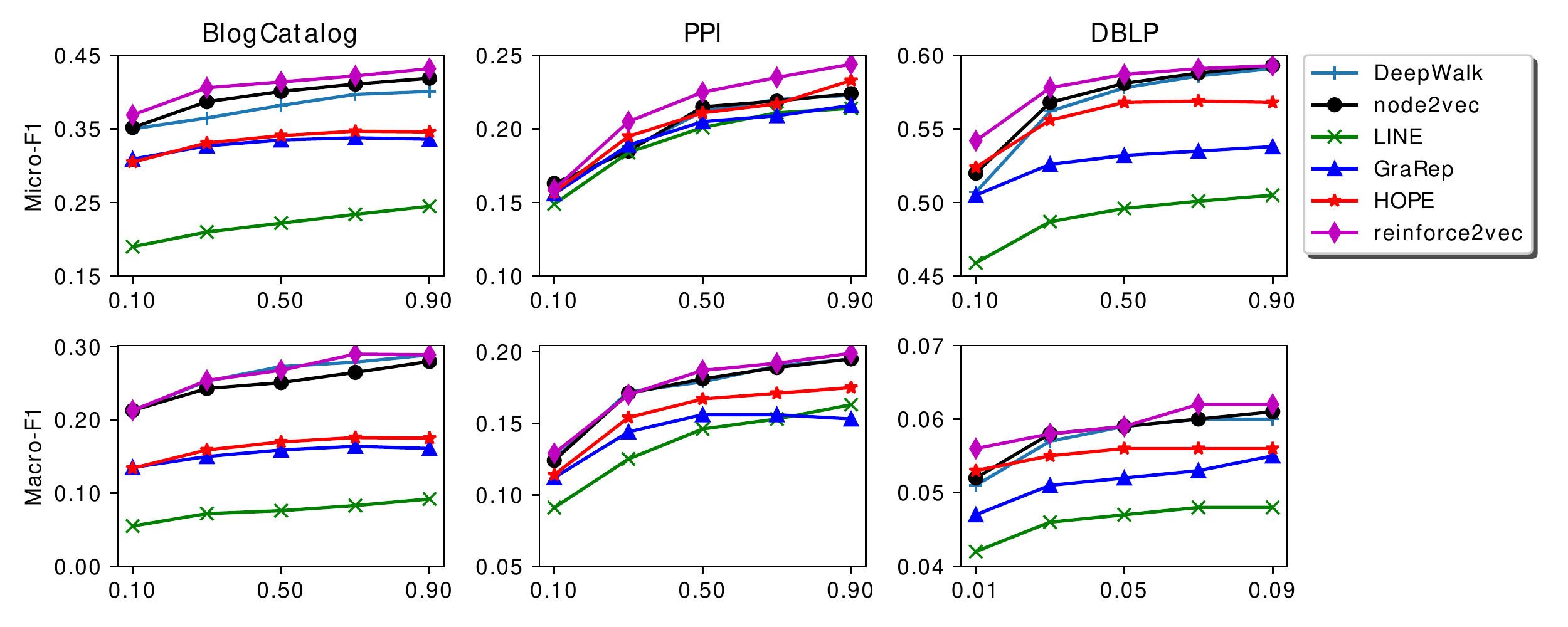}
\vspace{-0.1in}
\caption{Node classification task. Performance evaluation of different benchmarks on varying the amount of labeled data used for training. The x-axis denotes the fraction of labeled data, whereas the y-axis in the top and bottom rows denotes the Micro-F1 and Macro-F1 scores respectively.}
\label{fig:node-class}
\vspace{-0.2in}
\end{figure*}

\subsection{Node Classification Task.}
In the multi-label classification setting, every node is assigned one or more labels from a finite set. 
To conduct the experiment, we train reinforce2vec and baselines on the whole graph to obtain vertex representations and use a one-vs-rest logistic regression classifier as a classifier to perform node classification.
%During the training phase, we observe a certain fraction of nodes and all their labels and then use the embeddings of the labeled nodes as input features for a one-vs-rest logistic regression classifier. 

\subsubsection{Datasets.}
We use three datasets as follows for the task of node classification, details shown in Table \ref{tb:nc-data}.
\begin{itemize}
    \item BlogCatalog~\cite{zafarani2009social}: BlogCatalog is a network of social relationships of the bloggers.
    \item Protein-Protein Interactions (PPI)~\cite{breitkreutz2007biogrid}: PPI is a subgraph of the PPI network for Homo Sapiens. Node labels are extracted from hallmark gene sets and represent biological states.
    \item DBLP \cite{tang2008arnetminer}: DBLP is an academic citation network where authors are treated as nodes and their dominant conferences as labels.
\end{itemize}

\begin{table}[t]
    \centering
    \caption{Datasets for node classification.}
    \label{tb:nc-data}
    \begin{tabular}{c|ccc}
    \toprule
    Dataset & $\left | V \right |$ & $\left | E \right |$ & Labels \\
    \midrule
    Blogcatalog & 10,312 & 333,983 & 39 \\
    PPI & 3,890 & 76,584 & 50\\
    DBLP & 51,264 & 127,968 & 60 \\
    \bottomrule
    \end{tabular}
    \vspace{-0.2in}
\end{table}

%\vspace{-0.1in}
\subsubsection{Performance Evaluation.}
From Table \ref{tb:nc-result} we can see that our proposed algorithm consistently outperforms all the baselines on BlogCatalog, PPI, and DBLP. Specifically,  reinforce2vec achieves gains of 3.1\%, 4.7\%, and 1.0\% on Micro-F1 scores, compared with node2vec, the best baseline among three datasets. 
This indicates that the random walk applied in reinforce2vec is directly designed to optimize the equality of the whole path and can effectively encode the information of vertices into the learned representations.
%Moreover, we can improve the performances by approximate 1\% by adding the exploration score $U$. 
For a more fine-grained analysis, we also compare performance while varying the train-test split from 10\% to 90\% on both BlogCatalog and PPI datasets. For DBLP dataset, we split the dataset from 1\% to 9\% as training data since the performances change more obviously on less than 10\% data. The results for the Micro-F1 and Macro-F1 scores are shown in Figure \ref{fig:node-class}. We can also make similar observations. Overall, the performances on the task of node classification further prove that it is effective to consider a vertex-reinforced random walk for network embedding.

Furthermore, we evaluate the performance of each reinforced random walk based method using PPI dataset, shown in Table \ref{tb: rein-result}. We show Micro-F1 and Macro-F1 scores on the dataset with 50\% of the nodes labeled for training. We achieve consistent results about the effectiveness of the exploration strategy with the experiment on the link prediction task. A different experimental result is that: on node classification task, when calculating convergence score in DRRW, it seems that JS divergence based approach (reinforce2vec-d:JS) performs slightly better than KL based approach (reinforce2vec-d:KL). 
%The reason may be JS divergence is a symmetrized and smooth version of the KL divergence and more suitable for measuring the similarity between the distributions of two paths.  

\begin{table}[t]
\caption{Micro-F1 scores for node classification on BlogCatalog, PPI datasets with 50\% of the nodes labeled and DBLP with 5\% of the nodes labeled for training.}
\label{tb:nc-result}
\centering
\resizebox{0.45\textwidth}{!}{%
\begin{tabular}{c|ccc}
\toprule
Dataset & BlogCatalog & PPI & DBLP \\
 \midrule
DeepWalk & 38.2 & 21.2 & 57.8 \\
node2vec & 40.1 & 21.5 & 58.1 \\
LINE & 22.2 & 20.1 & 49.6 \\
GraRep & 33.5 & 20.5 & 53.2 \\
HOPE & 34.1 & 21.1 & 56.8 \\
reinforce2vec & \textbf{41.4} & \textbf{22.5} & \textbf{58.7} \\
\hline
Gain of reinforce2vec$\left [ \% \right ]$ & \textbf{3.1} & \textbf{4.7} & \textbf{1.0} \\
\bottomrule
\end{tabular}
}
 \vspace{-0.15in}
\end{table}

\begin{table}[t]
\centering
\caption{Performance for different variants of VRRW on node classification on PPI dataset with 50\% of the nodes labeled for training.}
\label{tb: rein-result}
\resizebox{0.45\textwidth}{!}{%
\begin{tabular}{c|c|cc}
\toprule
 \multicolumn{2}{c|}{Variant} & Micro-F1 & Macro-F1 \\
 \midrule
\multirow{3}{*}{reinforce2vec-v} & w/o exploration & 10.9 & 8.9 \\
 & $\epsilon$-Greedy & 19.0 & 16.2 \\
 & UCB & 18.9 & 16.3 \\
 \hline
\multirow{3}{*}{reinforce2vec-d:KL} & w/o exploration & 10.9 & 9.3 \\
 & $\epsilon$-Greedy & 21.4 & 17.5 \\
 & UCB & 22.0 & 18.2 \\
 \hline
\multirow{3}{*}{reinforce2vec-d:JS} & w/o exploration & 10.7 & 9.1 \\
 & $\epsilon$-Greedy & 21.9 & 18.3 \\
 & UCB & \textbf{22.5} & \textbf{18.6} \\
 \bottomrule
\end{tabular}
}
 \vspace{-0.15in}
\end{table}

\subsection{Parameter sensitivity.}
%Reinforce2vec involves a number of hyper-parameters. 
Here, we evaluate how different choices of $\epsilon$ in $\epsilon$-Greedy exploration affect the performance of reinforce2vec on Facebook dataset for link prediction and PPI dataset for node classification.  We test the influence of $\epsilon$  with $\epsilon= \left \{ 0, 0.1, 0.3, 0.5, 0.7, 0.9, 1 \right \}$. The results are shown in Figure \ref{fg:ps}. We get the best results when $\epsilon = 0.5$ on node classification task and $\epsilon = 0.3$ on link prediction task. When $\epsilon = 0$ (only exploitation), it is much worse which proves that the getting stuck problem would dramatically harm the performance for network embedding. Besides, when $\epsilon = 1$, reinforce2vec becomes DeepWalk naturally. As we can see, reinforce2vec outperforms DeepWalk.

%\begin{figure}[t]
%    \centering
%    \caption{Performance of baselines on varying the amount of labeled data used for training.}
%    \label{fg-performance}
%    \subfigure[BlogCatalog-Micro.]{\includegraphics[width=0.23\textwidth]{blog-micro.pdf}}
%    \subfigure[PPI-Micro.]{\includegraphics[width=0.23\textwidth]{ppi-micro.pdf}}
%    \subfigure[BlogCatalog-Macro.]{\includegraphics[width=0.23\textwidth]{blog-macro.pdf}}
%    \subfigure[PPI-Macro.]{\includegraphics[width=0.23\textwidth]{ppi-macro.pdf}}
%    \centering
%    \vspace{-5px}
%    \label{fig-online-compare}
%\end{figure}

%\subsection{Sensitivity Analysis}
%reinforce2vec-d involves several hyper-parameters. Here, we evaluate how different choices of walk length$L$ and trade-off $\lambda$ affect the performance of our model. Except for the parameter being tested, all other parameters are set as introduced in Section \ref{sec-setting}. 

%\begin{figure}[t]
%\begin{minipage}{0.58\textwidth} %
%\includegraphics[width=1\textwidth]{sa.pdf}
%    \vspace{-5px}
%    \caption{Sensitivity Aanalysis.}
%    \label{fig-online-compare}
%\end{minipage} %
%\begin{minipage}{.35\textwidth} %
%\includegraphics[width=1\textwidth]{ppi-micro.pdf}
%    \vspace{-5px}
%    \caption{Variation of node distribution during different random walks on PPI.}
%    \label{fig-online-compare}
%\end{minipage}
%\end{figure}

\begin{figure}[t]
    \centering
    \subfigure[Node classification.]{\includegraphics[width=0.23\textwidth]{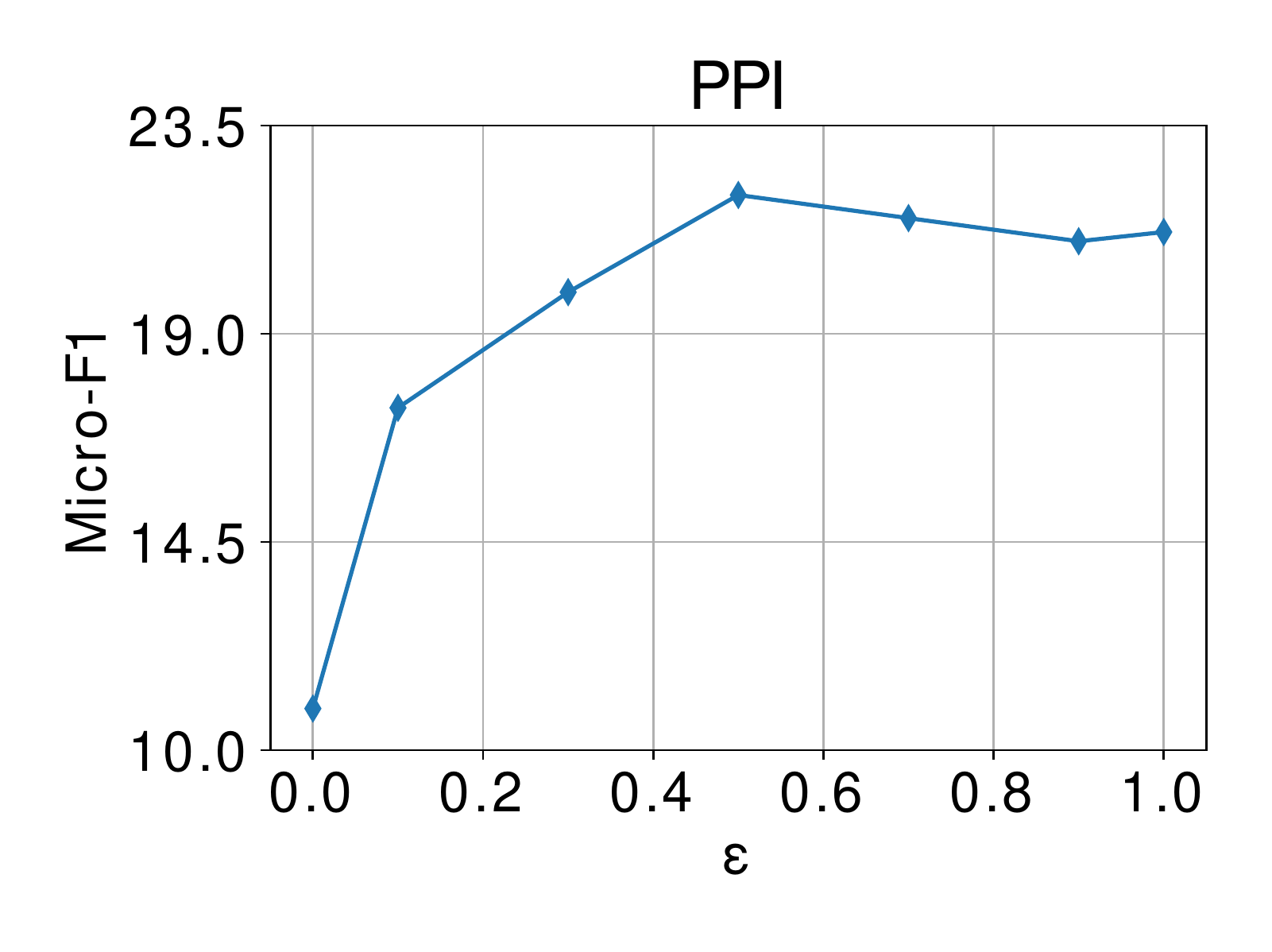}}
    \subfigure[Link Prediction.]{\includegraphics[width=0.23\textwidth]{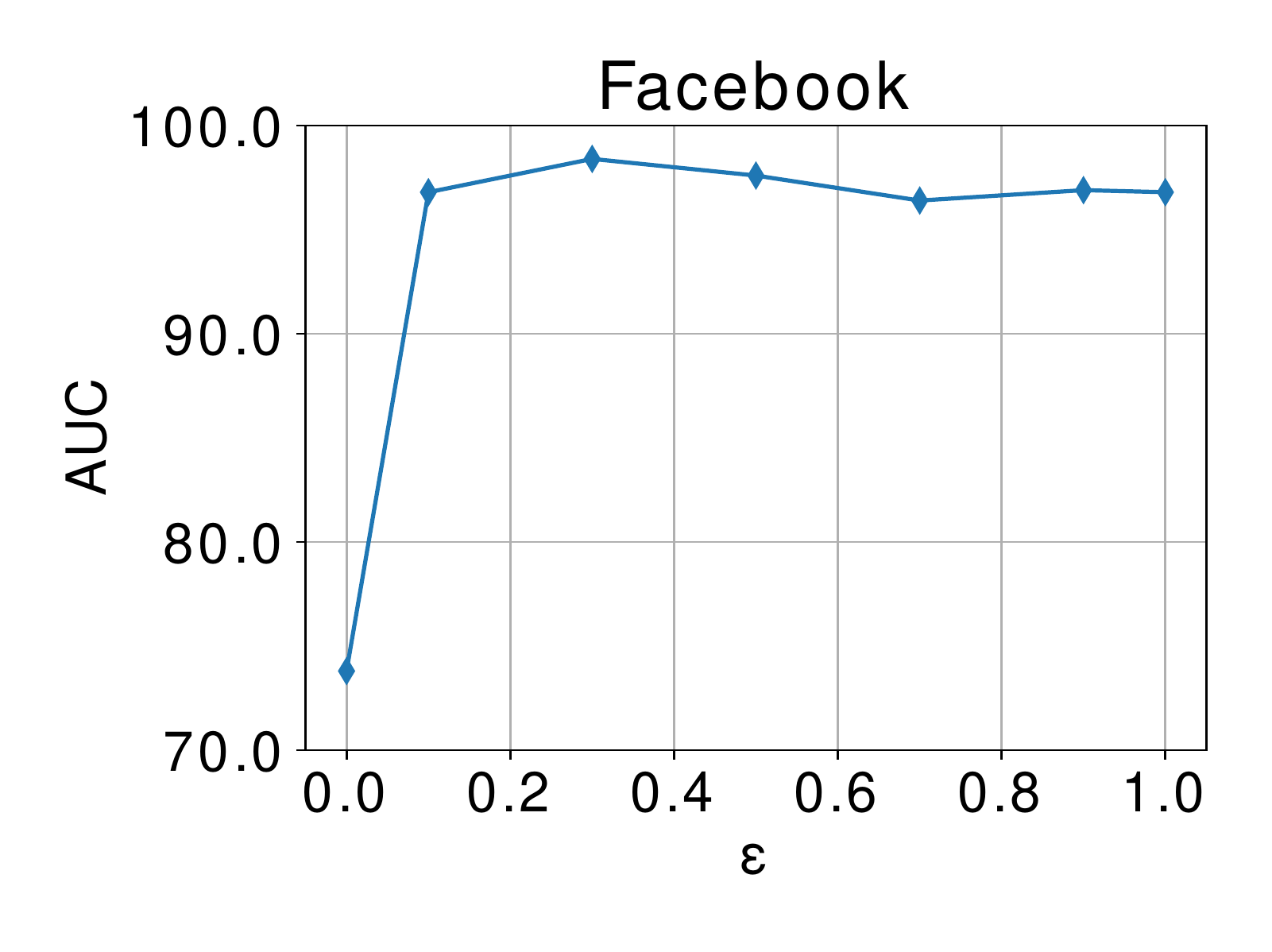}}
    \centering
    \vspace{-5px}
    \caption{Performance of reinforce2vec with $\epsilon$-Greedy on varying the $\epsilon$ on PPI dataset for node classification and on Facebook dataset for link prediction.}
    \label{fg:ps}
    \vspace{-0.2in}
    %\label{fig-online-compare}
\end{figure}

\section{Conclusion}
In this paper, we study the vertex-reinforced random walk (VRRW) for network embedding, which can make full use of the history of a random walk path. By using KL and JS divergences to guide the random walk, we also design a variant of VRRW, denoted as distribution-reinfoced random walk (DRRW).
To address the stuck problem of VRRW and DRRW, we further design an exploitation-exploration mechanism to help the random walk jump out of the stuck set.
By conducting extensive experiments for link prediction and node classification tasks, we demonstrate the effectiveness of the proposed reinforc2vec framework comparing to state-of-the-art random walk based embedding methods.
%\vspace{-0.1in}

\section*{Acknowlegement}
%We thank the anonymous reviewers for their valuable comments and suggestions that help improve the quality of this manuscript.
% This paper was done when the first author was an intern at WeBank AI Department.
This paper was supported by the HKUST-WeBank Joint Lab and the Early Career Scheme (ECS, No. 26206717) from Research Grants Council in Hong Kong.
% In this paper, we propose the reinfoce2vec framework for network embedding, which is based on a novel type of random walk. Different from previous on graphs, 
% which are Markovian chains, we design a distribution-reinforced random walk (reinforce2vec-d), which can not only converge efficiently to a stationary distribution but also alleviate the inherent ``stuck'' problem of VRRW. By applying reinforce2vec to the graph embedding tasks, reinforce2vec outperforms existing Markov based random walk methods. Moreover, reinforce2vec enjoys the advantages of automatically and efficiently stopping given a graph comparing to the fine-tuned longer walk length of Markov based random walk methods, leading to the huge practical potential for graph-based tasks.

%\clearpage
\bibliography{siam}
\bibliographystyle{siam}
%\section{Appendix}
%You may include other additional sections here. 

\end{document}

%% file: math_commands.tex
\usepackage{amsmath,amsfonts,bm}

\def\eqref#1{equation~\ref{#1}}
\def\1{\bm{1}}

\def\vm{{\mathbf{m}}}

\def\vu{{\mathbf{u}}}
\def\vv{{\mathbf{v}}}
\def\vw{{\mathbf{w}}}
\def\vx{{\mathbf{x}}}

\def\vz{{\mathbf{z}}}

\DeclareMathAlphabet{\mathsfit}{\encodingdefault}{\sfdefault}{m}{sl}
\SetMathAlphabet{\mathsfit}{bold}{\encodingdefault}{\sfdefault}{bx}{n}

\def\sR{{\mathbb{R}}}